\documentclass[a4paper,UKenglish,cleveref, autoref, thm-restate]{lipics-v2021}

\title{Southwest Tree: A Low-Memory Data Structure for Partial Accumulations by Non-Commutative Invertible Operations} 

\titlerunning{Southwest Tree} 

\author{Nicholas J.C. Papadopoulos\footnote{Corresponding author; Nicholas.Papadopoulos@colorado.edu}}{Department of Computer Science, University of Colorado Boulder, Boulder, CO
80309, USA \and \url{https://www.nicholas-papadopoulos.edu/} }{Nicholas.Papadopoulos@colorado.edu}{https://orcid.org/0000-0002-6357-0030}{}

\authorrunning{N. Papadopoulos} 

\Copyright{Nicholas J.C. Papadopoulos} 



\category{} 

\relatedversion{} 

\nolinenumbers 

\EventEditors{John Q. Open and Joan R. Access}
\EventNoEds{2}
\EventLongTitle{42nd Conference on Very Important Topics (CVIT 2016)}
\EventShortTitle{CVIT 2016}
\EventAcronym{CVIT}
\EventYear{2016}
\EventDate{December 24--27, 2016}
\EventLocation{Little Whinging, United Kingdom}
\EventLogo{}
\SeriesVolume{42}
\ArticleNo{23}

\begin{document}

\maketitle

\begin{abstract}
The task of accumulating a portion of a list of values, whose values may be updated at any time, is widely used throughout various applications in computer science.
While it is trivial to accomplish this task without any constraints, trivial solutions often sacrifice time complexity in either accumulating or updating the values, one being constant time and the other being linear.
To even out the complexity, two well-known data structures have been used to accomplish this task, namely the Segment Tree and the Binary Indexed Tree, which are able to carry out both tasks in $O(\log_2 N)$ time for a list of $N$ elements.
However, the Segment Tree suffers from requiring auxiliary memory to contain additional values, while the Binary Indexed Tree is unable to handle non-commutative accumulation operations.
Here, we present a data structure, called the Southwest Tree, that accomplishes these tasks for non-commutative, invertible accumulation operations in $O(\log_2 N)$ time and uses no additional memory to store the structure apart from the initial input array.
\end{abstract}

\section{Introduction}
The task of accumulating values frequently arises in everyday programming tasks. Depending on the case, one may find themselves wanting to be able to update one of the values at any time while being able to accumulate either all or a portion of them. A common example of this is called the ``Partial Sums'' problem, where the update to the elements is addition and the accumulation is the sum of the elements~\cite{Fredman1982TheCO}. 

\begin{definition}[Partial Sums]
Maintain an array of $N$ elements $I[1], I[2], \dots, I[N]$ that supports the following operations:
\begin{itemize}
    \item $\texttt{update}(i, \delta)$: set $I[i] = I[i] + \delta$
    \item $\texttt{accumulate}(i)$: return $\sum_{j=1}^{i} I[j]$
\end{itemize}
\end{definition}

Naive solutions of either keeping the initial array unmodified or storing the partial accumulation up to each index in the array can perform only one of these operations in $O(1)$ time but the other in $O(N)$ time. 
This topic has spawned data structures that better optimize these tasks so that they may both be accomplished in $O(\log_2 N)$ time.

The current, commonly utilized data structures for this task are the Segment Tree and the Binary Indexed Tree~\cite{deBerg2008,Fenwick1994AND}.
These both successfully perform the \texttt{update} and \texttt{accumulate} operations in $O(\log_2 N)$ time but differ in required memory.
On one hand, the Segment Tree requires up to $2N - 1$ array elements, while the Binary Indexed Tree requires only $N$ elements.

Considering different possible accumulation operations reveals another difference between the two solutions.
While Segment Trees are able to handle non-commutative operations, Binary Indexed Trees are not.
Since the $\texttt{sum}$ operation is commutative, both solutions can be used.
However, in the case of non-commutative operations, such as matrix multiplication, Binary Indexed Trees cannot be used without using additional memory.

This paper proposes a tree, called the Southwest Tree, that is stored using $O(N)$ memory, applies the \texttt{update} and \texttt{accumulate} operations in $O(\log_2 N)$ time, and can handle non-commutative, invertible accumulation operations. 
This, then, allows the same memory benefit as the Binary Indexed Tree over the Segment Tree without sacrificing run time and also allowing a wider set of accumulation operations.

This paper first outlines the general structure, layout, and properties of the Southwest Tree by defining variables, defining parent/child relationships, defining the values of each node, and providing visualizations of an example.
It then describes how to construct the tree and perform the \texttt{update} and \texttt{accumulate} operations while providing Python code samples to do so.
Finally, the time and space complexity is analyzed and compared to that of the Segment Tree.

\section{Related Work}

One data structure in use for this problem is the Segment Tree, which works by setting the initial values as leaf nodes of the tree and propagating upwards, combining two adjacent nodes into their parent~\cite{deBerg2008}. 
This adds nodes to store combined values in addition to the initial values themselves.
This does not affect time complexity at large scales, but it does require up to nearly twice as much memory as other options, which could be more of a detriment for certain memory-limited applications or operations. 
This data structure can be constructed in $O(N)$ time, perform the update and accumulate operations in $O(\log_2 N)$ time, and requires up to $2N - 1$ nodes.

The Binary Indexed Tree, or Fenwick Tree, is the other leading option for this problem~\cite{Fenwick1994AND}. 
It works by finding children and parents by simply incrementing or decrementing the last set bit in a node index's binary representation, while storing the accumulation from (including) the parent index to (excluding) it's own index of the initial array values. 
This saves on memory space compared to the Segment Tree. 
Hence, this data structure can be constructed in $O(N)$ time, perform the update and accumulate operations in $O(\log_2 N)$ time, and requires $N$ nodes.

Red-Black Trees have been used as an alternative to the Binary Indexed Tree~\cite{cormen01introduction,accumulation_tree}.
However, this solution, as with the Binary Indexed Tree, also cannot handle non-commutative accumulation operations. Furthermore, it requires additional time complexity for construction, namely $O(N \log_2 N)$.

Dietz, later modified by Ramen, Ramen, and Rao, proposes a structure that runs the \texttt{update} and \texttt{accumulate} operations in $\Theta(\log_2 N / \log_2 \log_2 N)$ amortized time~\cite{Dietz1989OptimalAF,ramen-ramen-rao}. 
It works by storing ``the list at the leaves of a nearly complete tree of branching factor b = $\Theta(\log ^\epsilon n)$,'' where $\epsilon$ is a positive constant less than one. 
It also uses a two-array scheme and precomputed tables for internal nodes to track children in constant amortized time. 
However, this scheme requires auxiliary memory greater than a single array of $N$ elements, and each internal node has worst-case time of $O(\log_2 N)$ to \texttt{update}.

Hon, Sadakane, and Sung review solutions to the similar \textit{Searchable Partial Sum} problem, which includes the capability of the \textit{search} opertaion, where $search(j)$ returns the smallest $i$ such that $sum(i) \geq j$~\cite{10.1007/978-3-540-24587-2_52}. 
These solutions, while interesting, do not provide improvements on the operations of the \textit{Partial Sum} problem.

\section{Structure and Layout}

The following variables will be defined here and referenced throughout the paper. 
$I$ is the initial array of $N$ elements of which we want to keep track, which will be modified in-place to represent the Southwest Tree that supports the \texttt{update} and $\texttt{accumulate}$ operations. 
We will also use $A$ as an variable name representing the in-place modified values of $I$ according to the construction of the Southwest Tree.
These two variables represent the same memory address but are useful distinctions in proofs and explanations.
Array indexing begins at 1 in visuals and text descriptions for ease of understanding, but it begins at 0 in code examples for slightly improved memory optimization. 
We define the $+$ operation on the value types stored in $I$ as a suitable accumulation operation, including non-commutative and invertible, such as matrix multiplication on invertible matrices.
We still refer to this operation as a sum throughout the paper, as it makes the use of traditional mathematical symbols, such as $\sum$, available.

$\texttt{update}(i, d)$ adds $d$ to $I[i]$, and $\texttt{accumulate}(i)$ returns the sum of the first $i$ elements in $I$. 
$i$ will be assumed to be a valid index in the tree, i.e., no less than 1 and no greater than $N$ in one-based indexing.
Phantom nodes are nodes that do not exist in memory but whose indices are used for tree traversal. 
These are described in more detail in Sec.~\ref{subsec:building}.

This paper refers to an ongoing example where $I$ initially consists of the odd numbers from 1 to 17, inclusive, as shown in Fig.~\ref{fig:initial}.
These elements will be theoretically arranged in the Southwest Tree as shown in Fig.~\ref{fig:sw}(a), but are actually stored by modifying $I$ to reflect the tree, as shown in Fig.~\ref{fig:sw}(b).
Note that the code examples throughout this paper build upon each other, so assume that each code example has access to the code listed in previous examples.

\begin{figure}
    \centering
    \includegraphics[width=.7\textwidth]{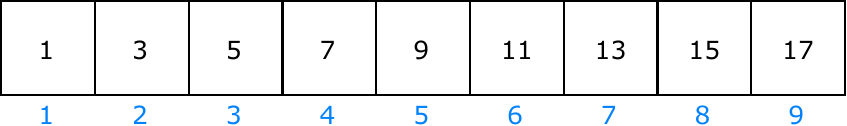}
    \caption{The array of initial elements. Element values are in black and element indices are in blue.}
    \label{fig:initial}
\end{figure}

\begin{figure}
    \centering
    \includegraphics[width=\textwidth]{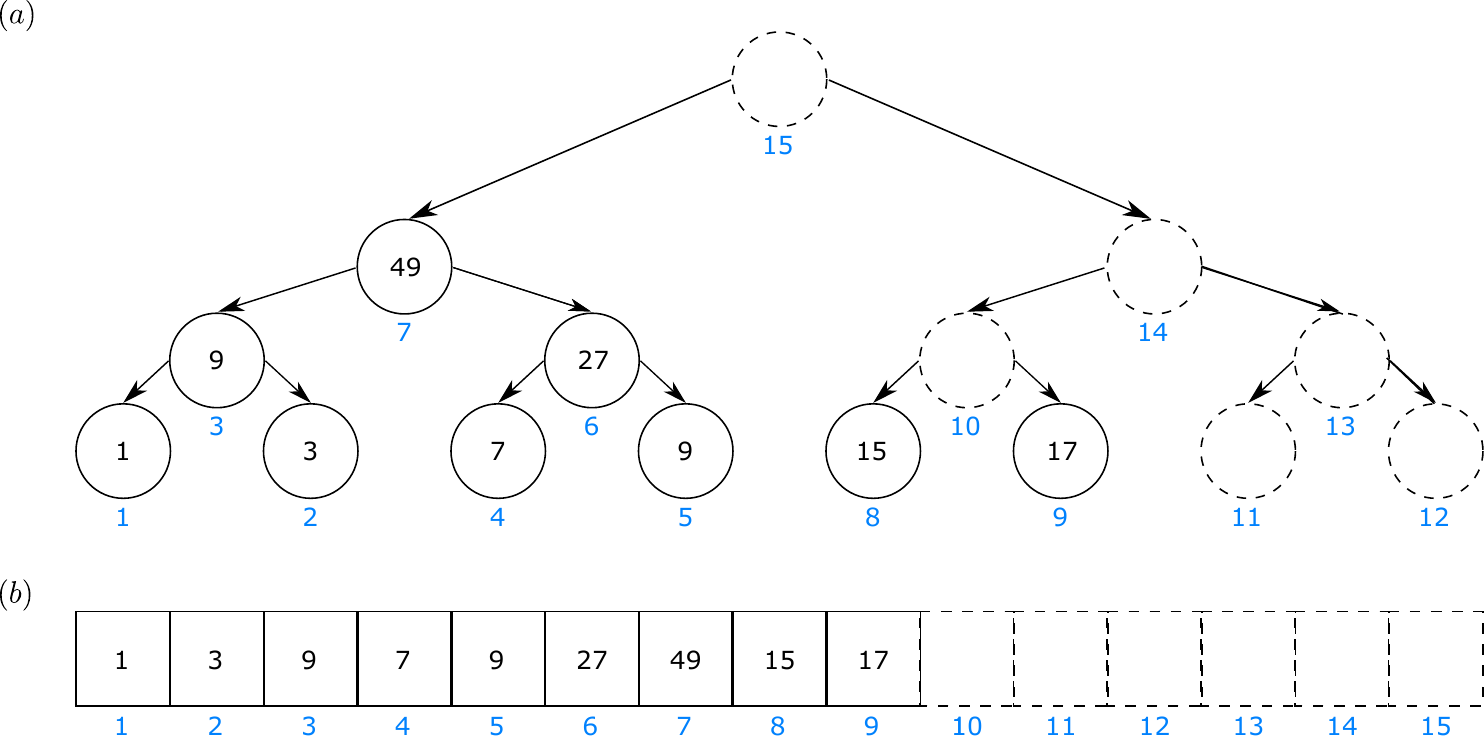}
    \caption{The Southwest Tree tracking initial elements $[1, 3, 5, 7 ,9, 11, 13, 15, 17]$ in the form of (a) a tree and (b) an array. Element values are in black and element indices are in blue. Dashed borders indicate phantom elements, which do not exist in memory.}
    \label{fig:sw}
\end{figure}

\subsection{Node Children}

Say some node $n$ has index $i$.
The overall structure of this tree is that the index of any node to the left of or below $n$ will be less than $i$.
This is the origin of the name ``Southwest Tree''.
To achieve this, the structure sets the right child with index $i - 1$ and the left child with index $i - 2^{h_i - 1}$, where $h_i$ is the height of $n$.
\begin{align}
    \texttt{left\_child}(i) &= i - 2^{h_i - 1} \label{eq:left-child} \\
    \texttt{right\_child}(i) &= i - 1 \label{eq:right-child}
\end{align}

For example, the node at index 7 in Fig.~\ref{fig:sw}(a) has left and right children at the following indices:
\begin{equation}
    \begin{aligned}
        \texttt{left\_child}(7) &= 7 - 2^{3 - 1} = 3 \\
        \texttt{right\_child}(7) &= 7 - 1 = 6.
    \end{aligned}
\end{equation}
Note that parent/child relationships remain the same regardless of tree size.
That is, no matter how many nodes are in the tree, the node at index 7, for example, will always have children at 3 and 6.

\subsection{Node Values}

The value of a node (except phantom nodes) at index $i$ is the value of its two children plus the value of $I[i]$,
\begin{equation}\label{eq:node-values}
    A[i] = A[i - 2^{h_i - 1}] + A[i - 1] + I[i].
\end{equation}
Since the tree is built bottom-up (see Sec.~\ref{subsec:building}), the child elements of $A[i]$ will already have the modified values of the tree.

For example, one can see the node at index 7 in Fig.~\ref{fig:sw}(a) is the sum of the values in nodes 3 and 6 plus the value of the initial element at index 7,
\begin{equation}
    \begin{aligned}
        A[7] &= A[3] + A[6] + I[7] \\
        &= 9 + 27 + 13 \\
        &= 49.
    \end{aligned}
\end{equation}

\subsection{Building the Tree} \label{subsec:building}

The algorithm begins at the root index, which will be $2^{\lfloor \log_2 N \rfloor + 1} - 1$.
It then executes a depth-first traversal until it reaches a leaf node index, which will be when the node height is 1.
The node value of the leaf nodes will be the initial element at the node's index. 
The algorithm then traverses back up the tree, where each node value will be the sum of its children plus the initial element at the node's index.

Although the algorithm will traverse through indices up to the root index, it will only store values in the original array and therefore does not store values at extra indices. 
We call the indices that are traversed but greater than $N$ phantom indices, and we call nodes at these indices phantom nodes. 
That is, they serve to help traversal, but only in mathematical operations.

Listing~\ref{lst:build} shows a Python example of the build operation. 
Note that array indexing is zero-based here and that the only modification required to go from zero-based to one-based is to instead set $root\_index = (root\_left << 1) - 1$.

\noindent
\begin{minipage}{\linewidth}
\begin{lstlisting}[language=Python, label=lst:build, caption=A Python example of the build operation. The input array A is assumed to already be instantiated and populated with the values of I here. `root\_left' is the difference between the root index and its left child index.]
num_elements = len(A)
root_left = 1 << num_elements.bit_length() - 1
root_index = (root_left << 1) - 2

def build(i, left):
    if left == 1:
        return
    for child_diff in (left, 1):
        build(i - child_diff, left >> 1)
    combine_children(i, left)

def combine_children(i, left):
    if i < num_elements:
        A[i] = A[i - left] + A[i - 1] + A[i]

build(root_index, root_left)
\end{lstlisting}
\end{minipage}

One can see that nodes with an index greater than the size of the initial array are not stored at all. 
This is shown in Fig.~\ref{fig:sw}, where the phantom nodes and array elements, i.e., all those with index greater than 9 in this case, have no value.

\begin{lemma}[Values of the array]\label{lem:arr-val} 
Given an initial array of values, $I$, modified to an array $A$ according to Sec.~\ref{subsec:building}, the value at index $i$ will be $A[i] = \sum_{j = i - 2^{h_i} + 2}^{i} I[j]$ while maintaining the order of operands, where $h_i$ is the height of the node at index $i$.
\end{lemma}

\begin{proof}
Let $N(i, h_i)$ be the value stored at index i with its node having height $h_i$, and let $N(i, 1) = I[i]$ since leaf nodes hold only the initial value.
We begin by proving that all subtrees with a node at height 2 is $\sum_{j = i - 2}^i I[j]$, where $i$ is the index of the root node in the subtree.
\begin{equation}
    \begin{aligned}
        N(i, 2) &= N(\texttt{left\_child}(i), 1) + N(\texttt{right\_child}(i), 1) + I[i] \\
        &= N(i - 2^{h_{i} - 1}, 1) + N(i - 1, 1) + I[i] \\
        &= I[i - 2^{2 - 1}] + I[i - 1] + I[i] \\
        &= I[i - 2] + I[i - 1] + I[i] \\
        &= \sum_{j = i - 2}^i I[j].
    \end{aligned}
\end{equation}
Note here that the order of operands is maintained.

Any node with height 3 is then
\begin{equation}
    \begin{aligned}
        N(i, 3) &= N(\texttt{left\_child}(i), 2) + N(\texttt{right\_child}(i), 2) + I[i] \\
        &= N(i - 2^{h_{i} - 1}, 2) + N(i - 1, 2) + I[i] \\
        &= N(i - 2^{3 - 1}, 2) + N(i - 1, 2) + I[i] \\
        &= N(i - 4, 2) + N(i - 1, 2) + I[i] \\
        &= \left(\sum_{j = i - 4 - 2}^{i - 4} I[j]\right) + \left(\sum_{j = i - 1 - 2}^{i - 1} I[j]\right) + I[i] \\
        &= \left(\sum_{j = i - 6}^{i - 4} I[j]\right) + \left(\sum_{j = i - 3}^{i - 1} I[j]\right) + I[i] \\
        &= \sum_{j = i - 6}^{i} I[j], \\
    \end{aligned}
\end{equation}
noting again that order is maintained.

We begin to see the pattern, then, stemming from the fact that all nodes contain the accumulation of values from its left and right subtrees, along with its own index.
Both subtrees accumulate values from all indices beginning with the leftmost leaf up to its own index.
In general, for any node at index $i$ with height $h_i$,
\begin{equation}
    \begin{aligned}
        N(i, h_i) &= N(\texttt{left\_child}(i), h_i-1) + N(\texttt{right\_child}(i), h_i-1) + I[i] \\
        &= N(i - 2^{h_{i} - 1}, h_i-1) + N(i - 1, h_i-1) + I[i] \\
        &= \left(\sum_{j = i - \sum_{k = 1}^{h_i - 1} 2^k}^{i - 2^{h_i - 1}} I[j]\right) + \left(\sum_{j = i - 1 - \sum_{k = 1}^{h_i - 2} 2^k}^{i - 1} I[j]\right) + I[i] \\
        &= \left(\sum_{j = i - (2^{h_i} - 2)}^{i - 2^{h_i - 1}} I[j]\right) + \left(\sum_{j = i - 1 - (2^{h_i - 1} - 2)}^{i - 1} I[j]\right) + I[i] \\
        &= \left(\sum_{j = i - 2^{h_i} + 2}^{i - 2^{h_i - 1}} I[j]\right) + \left(\sum_{j = i - 2^{h_i - 1} + 1}^{i - 1} I[j]\right) + I[i] \\
        &= \sum_{j = i - 2^{h_i} + 2}^{i} I[j].
    \end{aligned}
\end{equation}
Note here that the leftmost leaf of the right subtree has one greater index than the root of the left subtree.

\end{proof}

\section{Operations}

The following descriptions use the variable $curr$ to track the index of the current node in the recursive traversal.
If the desired index $i$ is less than or equal to the index of the left child, then traverse into the left subtree. 
Otherwise, traverse into the right subtree. 
Recurse until you reach the node where $curr == i$.
The helper function in Listing~\ref{lst:getchild} finds the root index of the subtree into which the recursion should traverse.

\noindent
\begin{minipage}{\linewidth}
\begin{lstlisting}[language=Python, label=lst:getchild, caption=A Python example of a helper function that returns the root index of the left or right subtree into which the recursion should traverse. $i$ is the desired index and $left$ is the difference in index to the left child.]
def get_child(i, curr, left):
    diff = left if i <= curr - left else 1
    return curr - diff
\end{lstlisting}
\end{minipage}

\subsection{Accumulate}\label{sec:accumulate}

The $\texttt{accumulate}(i)$ operation, shown in Listing~\ref{lst:accumulate}, is performed by traversing the tree beginning at the root node. 
If $i$ is equal to $curr$, accumulate the value of $A[curr]$ to the total and break out of the traversal. 
If $i$ is less than or equal to the index of the left child, then traverse into the left subtree. 
Otherwise, accumulate the value of the left child to the total and traverse into the right subtree. 
There is no need to check for phantom nodes here, because if the left child is phantom, then $i$ must be less than its index, hence traversing into the tree with no attempt to accumulate a phantom value.

\noindent
\begin{minipage}{\linewidth}
\begin{lstlisting}[language=Python, label=lst:accumulate, caption=A Python example of the accumulate operation.]
def accumulate(i):
    return acc_rec(i, root_index, root_left)

def acc_rec(i, curr, left):
    if curr == i:
        return A[curr]
    val = acc_rec(i, get_child(i, curr, left), left >> 1)
    return A[curr-left] + val if i > curr-left else val
\end{lstlisting}
\end{minipage}

\begin{theorem}[Accumulation]\label{thrm:accumulation}
Given an initial array of values, $I$, modified to an array $A$ according to Sec.~\ref{subsec:building}, the function $\texttt{accumulate}(i)$ described in Sec.~\ref{sec:accumulate} correctly returns $\sum_{j=1}^{i} I[j]$ while maintaining the order of operands.
\end{theorem}

\begin{proof}
We can write the traversal of the path to a desired node $n$ with an array, where a 1 indicates a traversal into the left subtree, and a 0 into the right subtree. The index $a$ of the array indicates the node at height $h_r - a$, $h_r$ again being the height of the root node.
For example, a traversal in Fig.~\ref{fig:sw}(a) to index 5 would be $[1, 0, 0]$.
So, for some given array $T = [t_1, t_2, \dots, t_k]$, where $t_a$ indicates the $a$th element of $T$ and $k$ being the final index of $T$, we can say that the index $i$ of a node at traversal step $a$ with root index $2^{h_r} - 1$ is 
\begin{equation}
    \mathcal{I}(a) = 2^{h_r} - 1 - \sum_{j = 1}^{a} 2^{t_j(h_r - j)}
\end{equation}
due to Eqs.~\eqref{eq:left-child} and~\eqref{eq:right-child}.

The $\texttt{accumulate}(i, \delta)$ of the Southwest Tree adds nothing when traversing into the left subtree, while traversals into the right subtree add the value of the root of the left subtree.
Therefore, let us consider subarrays of $T$ consisting of consecutive $1$s followed by a single $0$, which dictate each operand from a rightward traversal.
Let the traversal step of traversals into right subtrees be $R = [r_1, r_2, \dots, r_\ell] = [a : 1 \leq a \leq k \wedge t_a = 0]$, where $r_b$ indicates the $b$th element of $R$ and $\ell$ is the final index of $R$.

Let $\mathcal{P}(i)$ be the parent of the node at index $i$, with $\mathcal{P}_r(i)$ being the parent of a right child, i.e., the left child of its parent.
By Eq.~\eqref{eq:right-child}, given index $i$ of a right child,
\begin{equation}\label{eq:parent-right}
    \begin{aligned}
        \mathcal{P}_r(i) &= i + 1.
    \end{aligned}
\end{equation}
Let $\mathcal{S}(i)$ be the index of the sibling of a right child at index $i$.
By Eqs.~\eqref{eq:left-child} and~\eqref{eq:parent-right},
\begin{equation}
    \begin{aligned}
        \mathcal{S}(i) &= i + 1 - 2^{h_{i + 1} - 1}. \\
    \end{aligned}
\end{equation}

The first $t_a = 0$ will be preceded by a sequence of zero or more $1$s, traversing left from the root node.
Hence, the first operand, detailed in Appendix~\ref{sec:exp-acc-first}, will be 
\begin{equation}\label{eq:acc-first}
    \begin{aligned}
        A[\mathcal{S}(\mathcal{I}(r_1))] &= \sum_{j = 1}^{\mathcal{S}(\mathcal{I}(r_1))} I[j].\\
    \end{aligned}
\end{equation}
The $m$th $t_a = 0$ will traverse left from the node at index $\mathcal{I}(r_{m - 1})$.
Hence, the $m$th operand, detailed in Appendix~\ref{sec:exp-acc-mth}, will be 
\begin{equation}\label{eq:acc-mth}
    \begin{aligned}
        A[\mathcal{S}(\mathcal{I}(r_m))] &= \sum_{j = \mathcal{S}(\mathcal{I}(r_{m-1})) + 1}^{\mathcal{S}(\mathcal{I}(r_m))} I[j].\\
    \end{aligned}
\end{equation}
Upon reaching the desired node $n$ at index $i$ after traversal direction $t_k$, detailed in Appendix~\ref{sec:exp-acc-final}, the final operand will be
\begin{equation}\label{eq:acc-final}
    A[i] = \sum_{j = \mathcal{S}(\mathcal{I}(r_\ell)) + 1}^{i} I[j].\\
\end{equation}
Therefore, the final accumulation will be 
\begin{equation}
    \begin{aligned}
        \texttt{accumulate}(i) &= \left(\sum_{j = 1}^{\mathcal{S}(\mathcal{I}(r_1))} I[j]\right) \\
        &\quad + \left(\sum_{m = 2}^{\ell} \sum_{j = \mathcal{S}(\mathcal{I}(r_{m - 1})) + 1}^{\mathcal{S}(\mathcal{I}(r_m))} I[j]\right)\\
        &\quad + \sum_{j = \mathcal{S}(\mathcal{I}(r_\ell)) + 1}^{i} I[j] \\
        &= \sum_{j = 1}^i I[j].
    \end{aligned}
\end{equation}

\end{proof}

\subsection{Update}\label{sec:update}

The $\texttt{update}(i, d)$ operation, shown in Listing~\ref{lst:update}, is performed by traversing the tree beginning at the root node.
On the way down, the left and right children are inverted so that the initial value at that index is retrieved.
When $i$ is reached, modify the value at that index.
On the way back up the traversal, set the value of each parent node in the same way as it was when initially built.

\noindent
\begin{minipage}{\linewidth}
\begin{lstlisting}[language=Python, label=lst:update, caption=A Python example of the update operation. The function $inv(x)$ here returns the inverse of $x$ with respect to the accumulation operation.]
def update(i, d):
    update_rec(i, d, root_index, root_left)

def update_rec(i, d, curr, left):
    if curr == i:
        A[i] += d
        return
    invert_children(curr, left)
    update_rec(i, d, get_child(i, curr, left), left >> 1)
    combine_children(curr, left)

def invert_children(i, left):
    if i < num_elements:
        A[i] = inv(A[i - 1]) + inv(A[i - left]) + A[i]
\end{lstlisting}
\end{minipage}

\begin{theorem}[Update]
Given an initial array of values, $I$, modified to an array $A$ according to Sec.~\ref{subsec:building}, the function $\texttt{update}(i, \delta)$ described in Sec.~\ref{sec:update} correctly sets $I[i] = I[i] + \delta$ while maintaining the order of operands.
\end{theorem}

\begin{proof}
    Say that $inv(A[c]) + A[c]$ for any index $c$ yields the identity of the given operator $+$.
    Each step of the traversal begins by inverting the values of the root node's children, leaving only $I[curr]$. 
    In order to maintain the order of non-commutative operands, we add the inverse of the children in reverse order from Eq.~\eqref{eq:node-values}, i.e., we set $A[c] = inv(A[c - 1]) + inv(A[c - 2^{h_c - 1}]) + A[c]$.
    \begin{equation}
        \begin{aligned}
            A[c] &= inv(A[c - 1]) + inv(A[c - 2^{h_c - 1}]) + A[c - 2^{h_c - 1}] + A[c - 1] + I[c] \\
            &= inv(A[c - 1]) + A[c - 1] + I[c] \\
            &= I[c].
        \end{aligned}
    \end{equation}

    When the desired index $i$ is reached, after $A[i]$ is inverted back to $I[i]$, we set $I[i] = I[i] + \delta$.
    Finally, we traverse back up the tree, setting the node values according to Eq.~\eqref{eq:node-values}.
    \begin{equation}
        \begin{aligned}
            A[i] &= A[i - 2^{h_i - 1}] + A[i - 1] + I[i] + \delta \\
            A[\mathcal{P}(i)] &= A[\mathcal{P}(i) - 2^{h_{\mathcal{P}(i)} - 1}] + A[\mathcal{P}(i) - 1] + I[\mathcal{P}(i)] \\
            A[\mathcal{P}(\mathcal{P}(i))] &= A[\mathcal{P}(\mathcal{P}(i)) - 2^{h_{\mathcal{P}(\mathcal{P}(i))} - 1}] + A[\mathcal{P}(\mathcal{P}(i)) - 1] + I[\mathcal{P}(\mathcal{P}(i))] \\
            &= \dots
        \end{aligned}
    \end{equation}
    In this manner, all order of operands will be maintained and $I[i]$ modified as intended, so that $\texttt{accumulate}(i)$ will continue to function as intended.
\end{proof}

\section{Complexity Comparison}

This section describes the time and space complexity for the Southwest Tree and highlights how it outperforms the Segment Tree, which is generally considered the leading solution for non-commutative partial accumulations.
The Southwest Tree performs with the same time complexity while reducing the space requirements.

\subsection{Time}

\subsubsection{Update and Accumulate}
These operations, in the worst case, traverse the entire tree and use constant time for each step.
Therefore, the time complexity for these operations will be dependent on the height of the tree, which, as described in Sec.~\ref{subsec:building}, is $\lfloor \log_2 (N) \rfloor + 1$ for the Southwest Tree, giving it $O(\log_2 N)$ time complexity. 
This matches the time complexity of these operations on the leading alternative of the Segment Tree.

\subsubsection{Build} \label{subsubsec:build-complexity}
To build the tree, one simply needs to traverse each node twice, and the tree will have $2^{\lfloor \log_2 N \rfloor + 1} - 1$ nodes. 
The best case is when $N = 2^x - 1$ for any integer $x \in \mathbb{N}$, where $\mathbb{N}$ is the set of natural numbers.
In this case, the tree will have $2^{\lfloor \log_2 (2^x - 1) \rfloor + 1} - 1 = 2^{x - 1 + 1} - 1 = 2^x - 1 = N$, having zero phantom nodes.
In the worst case, $N = 2^x$ for any integer $x \in \mathbb{N}$.
In this case, the tree will have $2^{\lfloor \log_2 (2^x) \rfloor + 1} - 1 = 2^{x + 1} - 1 = 2N - 1$, having $2N - 1 - N = N - 1$ phantom nodes. 

Each node only implements constant time operations to set its values, so in either case the build time for the Southwest Tree is $O(N)$, matching the build time complexity of the Segment Tree.

\subsection{Space}
Since both the Southwest Tree and Segment Tree recursively traverse the tree to the bottom in the worst case, they equally require $O(\log_2 N)$ stack space during operations.
However, although both of these trees have $O(N)$ space complexity to store the values, the Southwest Tree modifies the input array in-place and does not require any additional space for values.
This beats the Segment Tree, which needs $O(N)$ additional space to store values on top of the initial array.

\section{Conclusion}

This paper proposes a tree structure, the Southwest Tree, of $N$ nodes that can be built in $O(N)$ time and handle \texttt{update} and \texttt{accumulate} functions on an input array with $O(\log_2 N)$ time complexity. 
It does this by utilizing the principle that nodes to the left of or below some given node have lower indices than the given node. 
It also uses phantom nodes for traversal, allowing it to conserve space but maintain the functionality of a full tree. 

Currently, the most popular solutions to handle these operations are the Segment Tree and the Binary Indexed Tree. 
However, since the Binary Indexed Tree cannot handle non-commutative accumulation operations, the Segment Tree is usually used instead, requiring additional memory.
The Southwest Tree therefore improves on both of these solutions by accurately handling non-commutative, invertible accumulation operations while only using $N$ array elements for storage.



\appendix

\section{Expanded Derivation of Eq.~\ref{eq:acc-first}}\label{sec:exp-acc-first}
The first $t_a = 0$ will be preceded by a sequence of zero or more $1$s, traversing left from the root node.
Hence, the first operand, by Lemma~\ref{lem:arr-val}, will be 
\begin{equation}
    \begin{aligned}
        A[\mathcal{S}(\mathcal{I}(r_1))] &= \sum_{j = \mathcal{S}(\mathcal{I}(r_1)) - 2^{h_r - r_1} + 2}^{\mathcal{S}(\mathcal{I}(r_1))} I[j]\\
        &= \sum_{j = \mathcal{I}(r_1) + 1 - 2^{h_{\mathcal{I}(r_1) + 1} - 1} - 2^{h_r - r_1} + 2}^{\mathcal{S}(\mathcal{I}(r_1))} I[j]\\
        &= \sum_{j = \mathcal{I}(r_1) + 1 - 2^{h_{\mathcal{P}_r(\mathcal{I}(r_1))} - 1} - 2^{h_r - r_1} + 2}^{\mathcal{S}(\mathcal{I}(r_1))} I[j]\\
        &= \sum_{j = \mathcal{I}(r_1) + 1 - 2^{h_r - r_1 + 1 - 1} - 2^{h_r - r_1} + 2}^{\mathcal{S}(\mathcal{I}(r_1))} I[j]\\
        &= \sum_{j = \mathcal{I}(r_1) + 1 - 2^{h_r - r_1} - 2^{h_r - r_1} + 2}^{\mathcal{S}(\mathcal{I}(r_1))} I[j]\\
        &= \sum_{j = \mathcal{I}(r_1) + 1 - 2^{h_r - r_1 + 1} + 2}^{\mathcal{S}(\mathcal{I}(r_1))} I[j]\\
        &= \sum_{j = 2^{h_r} - 1 - \left(\sum_{j = 1}^{r_1} 2^{t_j(h_r - j)}\right) + 1 - 2^{h_r - r_1 + 1} + 2}^{\mathcal{S}(\mathcal{I}(r_1))} I[j]\\
        &= \sum_{j = 2^{h_r} - 1 - \left(\sum_{j = 1}^{r_1 - 1} 2^{h_r - j}\right) - 2^{h_r - r_1 + 1} + 2}^{\mathcal{S}(\mathcal{I}(r_1))} I[j]\\
        &= \sum_{j = 2^{h_r} - 1 - 2^{h_r - r_1 + 1 + (r_1 - 1 - 1 + 1)} + 2}^{\mathcal{S}(\mathcal{I}(r_1))} I[j]\\
        &= \sum_{j = 2^{h_r} - 1 - 2^{h_r} + 2}^{\mathcal{S}(\mathcal{I}(r_1))} I[j]\\
        &= \sum_{j = 1}^{\mathcal{S}(\mathcal{I}(r_1))} I[j]\\
    \end{aligned}
\end{equation}
where we have used 
\begin{equation}
    \begin{aligned}
        \left(\sum_{j = 1}^{r_1} 2^{t_j(h_r - j)}\right) + 1 &= \left(\sum_{j = 1}^{r_1 - 1} 2^{1(h_r - j)}\right) - 2^{0(h_r - j)} + 1 \\
        &= \left(\sum_{j = 1}^{r_1 - 1} 2^{h_r - j}\right) - 1 + 1 \\
        &= \sum_{j = 1}^{r_1 - 1} 2^{h_r - j},
    \end{aligned}
\end{equation}
because $t_j = 1 \quad \forall 0 \leq j \leq h_r - r_1 - 1$ and $t_{r_1} = 0$, and 
\begin{equation}
    \begin{aligned}
        \left(\sum_{j = 1}^{r_1 - 1} 2^{h_r - j}\right) + 2^{h_r - r_1 + 1} &= 2^{h_r - r_1 + 1 + (r_1 - 1 - 1 + 1)}
    \end{aligned}
\end{equation}
because the summation $\sum_{j = 1}^{r_1 - 1} 2^{h_r - j}$ adds the number of operands of the summation, $(r_1 - 1) - (1) + 1$, to the exponent of $2^{h_r - r_1 + 1}$, as one can see by combining the operands one by one as follows: 
\begin{equation}
    \begin{aligned}
        \left(\sum_{j = 1}^{r_1 - 1} 2^{h_r - j}\right) + 2^{h_r - r_1 + 1} &= \left(\sum_{j = 1}^{r_1 - 2} 2^{h_r - j}\right) + 2^{h_r - r_1 + 2}\\
        &= \left(\sum_{j = 1}^{r_1 - 3} 2^{h_r - j}\right) + 2^{h_r - r_1 + 3}\\
        &= \dots \\
    \end{aligned}
\end{equation}

\section{Expanded Derivation of Eq.~\ref{eq:acc-mth}}\label{sec:exp-acc-mth}
The $m$th $t_a = 0$ will traverse left from the node at index $\mathcal{I}(r_{m - 1})$.
Hence, the $m$th operand, following similar logic to Sec.~\ref{sec:exp-acc-first}, will be 
\begin{equation}
    \begin{aligned}
        A[\mathcal{S}(\mathcal{I}(r_m))] &= \sum_{j = \mathcal{S}(\mathcal{I}(r_m)) - 2^{h_r - r_m} + 2}^{\mathcal{S}(\mathcal{I}(r_m))} I[j]\\
        &= \sum_{j = \mathcal{I}(r_m) + 1 - 2^{h_{\mathcal{I}(r_m) + 1} - 1} - 2^{h_r - r_m} + 2}^{\mathcal{S}(\mathcal{I}(r_m))} I[j]\\
        &= \sum_{j = \mathcal{I}(r_m) + 1 - 2^{h_{\mathcal{P}_r(\mathcal{I}(r_m))} - 1} - 2^{h_r - r_m} + 2}^{\mathcal{S}(\mathcal{I}(r_m))} I[j]\\
        &= \sum_{j = \mathcal{I}(r_m) + 1 - 2^{h_r - r_m + 1 - 1} - 2^{h_r - r_m} + 2}^{\mathcal{S}(\mathcal{I}(r_m))} I[j]\\
        &= \sum_{j = \mathcal{I}(r_m) + 1 - 2^{h_r - r_m} - 2^{h_r - r_m} + 2}^{\mathcal{S}(\mathcal{I}(r_m))} I[j]\\
        &= \sum_{j = \mathcal{I}(r_m) + 1 - 2^{h_r - r_m + 1} + 2}^{\mathcal{S}(\mathcal{I}(r_m))} I[j]\\
        &= \sum_{j = 2^{h_r} - 1 - \left(\sum_{j = 1}^{r_m} 2^{t_j(h_r - j)}\right) + 1 - 2^{h_r - r_m + 1} + 2}^{\mathcal{S}(\mathcal{I}(r_m))} I[j]\\
        &= \sum_{j = 2^{h_r} - 1 - \left(\sum_{j = 1}^{r_{m-1}} 2^{t_j(h_r - j)}\right) - \left(\sum_{j = r_{m-1} + 1}^{r_m} 2^{t_j(h_r - j)}\right) + 1 - 2^{h_r - r_m + 1} + 2}^{\mathcal{S}(\mathcal{I}(r_m))} I[j]\\
        &= \sum_{j = \mathcal{I}(r_{m-1}) - \left(\sum_{j = r_{m-1} + 1}^{r_m} 2^{t_j(h_r - j)}\right) + 1 - 2^{h_r - r_m + 1} + 2}^{\mathcal{S}(\mathcal{I}(r_m))} I[j]\\
        &= \sum_{j = \mathcal{I}(r_{m-1}) - \left(\sum_{j = r_{m-1} + 1}^{r_m - 1} 2^{h_r - j}\right) - 2^{h_r - r_m + 1} + 2}^{\mathcal{S}(\mathcal{I}(r_m))} I[j] \\
        &= \sum_{j = \mathcal{I}(r_{m-1}) - 2^{h_r - r_m + 1 + (r_m - 1 - (r_{m-1} + 1) + 1)} + 2}^{\mathcal{S}(\mathcal{I}(r_m))} I[j]\\
        &= \sum_{j = \mathcal{I}(r_{m-1}) - 2^{h_r - r_{m-1}} + 2}^{\mathcal{S}(\mathcal{I}(r_m))} I[j]\\
        &= \sum_{j = \mathcal{S}(\mathcal{I}(r_{m-1})) + 1}^{\mathcal{S}(\mathcal{I}(r_m))} I[j].\\
    \end{aligned}
\end{equation}

\section{Expanded Derivation of Eq.~\ref{eq:acc-final}}\label{sec:exp-acc-final}
Upon reaching the desired node $n$ at index $i$ after traversal direction $t_k$, by Lemma~\ref{lem:arr-val},
\begin{equation}
    \begin{aligned}
        A[i] &= \sum_{j = i - 2^{h_i} + 2}^{i} I[j]\\
        &= \sum_{j = \mathcal{I}(k) - 2^{h_r - k} + 2}^{i} I[j].\\
    \end{aligned}
\end{equation}
If $t_k = 0$, then $r_\ell = k$ and $\mathcal{I}(k) = \mathcal{S}(\mathcal{I}(r_\ell)) + 2^{h_r - r_\ell} - 1$ and
\begin{equation}
    \begin{aligned}
        A[i] &= \sum_{j = \mathcal{I}(k) - 2^{h_r - k} + 2}^{i} I[j]\\
        &= \sum_{j = \mathcal{S}(\mathcal{I}(r_\ell)) + 2^{h_r - r_\ell} - 1 - 2^{h_r - k} + 2}^{i} I[j]\\
        &= \sum_{j = \mathcal{S}(\mathcal{I}(r_\ell)) + 1}^{i} I[j].\\
    \end{aligned}
\end{equation}
If $t_k = 1$, then $\mathcal{I}(k) = \mathcal{I}(r_\ell) - \sum_{j = r_\ell + 1}^{k} 2^{h_r - j}$ and
\begin{equation}
    \begin{aligned}
        A[i] &= \sum_{j = \mathcal{I}(k) - 2^{h_r - k} + 2}^{i} I[j]\\
        &= \sum_{j = \mathcal{I}(r_\ell) - \left(\sum_{j = r_\ell + 1}^{k} 2^{h_r - j}\right) - 2^{h_r - k} + 2}^{i} I[j]\\
        &= \sum_{j = \mathcal{I}(r_\ell) - 2^{h_r - k + (k - (r_\ell + 1) + 1)} + 2}^{i} I[j]\\
        &= \sum_{j = \mathcal{I}(r_\ell) - 2^{h_r - r_\ell} + 2}^{i} I[j]\\
        &= \sum_{j = \mathcal{S}(\mathcal{I}(r_\ell)) + 1}^{i} I[j].\\
    \end{aligned}
\end{equation}
Therefore, both cases end with the final equation
\begin{equation}
    A[i] =\sum_{j = \mathcal{S}(\mathcal{I}(r_\ell)) + 1}^{i} I[j].\\
\end{equation}

\end{document}